\documentclass[11pt]{article}
\usepackage[margin=1in]{geometry}
\usepackage{cite,microtype,graphicx}
\usepackage{hyperref,aliascnt}
\usepackage{amsmath,amssymb,amsfonts,amsthm}

\newtheorem{theorem}{Theorem}

\newaliascnt{lemma}{theorem} \newtheorem{lemma}[lemma]{Lemma}
\aliascntresetthe{lemma} 

\newaliascnt{corollary}{theorem}
\newtheorem{corollary}[corollary]{Corollary}
\aliascntresetthe{corollary}

\newaliascnt{definition}{theorem}
\newtheorem{definition}[definition]{Definition}
\aliascntresetthe{definition}

\newcommand{\sharpp}{\(\#\mathsf{P}\)}
\newcommand{\NP}{\(\mathsf{NP}\)}

\title{Counting Polygon Triangulations is Hard}

\author{David Eppstein\\
Computer Science Department, University of California, Irvine}

\date{ }

\begin{document}
\maketitle  

\begin{abstract}
We prove that it is $\#\mathsf{P}$-complete to count the triangulations of a (non-simple) polygon.
\end{abstract}

\section{Introduction}

In 1979, Leslie Valiant published his proof that it is \sharpp-complete to compute the permanent of a 0--1 matrix, or equivalently to count the perfect matchings of a bipartite graph~\cite{Val-TCS-79}.\footnote{See \autoref{sec:sharpp} for the definitions of \sharpp{} and \sharpp-completeness.} This result was significant in two ways: It was surprising at the time that easy (polynomial time) existence problems could lead to intractable counting problems, and it opened the door to hardness proofs for many other counting problems, primarily in graph theory.
Beyond individual problems, several broad classifications of hard graph counting problems are now known. For instance, it is \sharpp-complete to count $k$-colorings of a graph or determine the value of its Tutte polynomial at any argument of the polynomial outside of a small finite set of exceptions~\cite{JaeVerel-MPCPS-90},
and all problems of counting homomorphisms to a fixed directed acyclic graph are either \sharpp-complete or polynomial~\cite{DyeGolPat-JACM-07}.

However, in computational geometry, only a small number of isolated problems have been shown to be \sharpp-complete or \sharpp-hard. These include counting the vertices or facets of high-dimensional convex polytopes~\cite{Lin-SIDMA-86},  computing the expected total length of the minimum spanning tree of a stochastic subset of three-dimensional points~\cite{KamSur-ANALCO-11}, and counting linear extensions of two-dimensional dominance partial orderings~\cite{DitPak-18}. Because of the rarity of complete problems in this area, other problems of counting geometric structures can only be inferred to be \NP-hard, in cases for which constructing the structure is \NP-hard~\cite{AlvBriCur-DCG-15,AsaAsaIma-JACM-86,SeaPil-EuroCG-17,Lub-SoCG-85,KanBod-IC-97}. There has been significant research on counting easy-to-construct non-crossing configurations in the plane, including matchings, simple polygons, spanning trees,  triangulations, and pseudotriangulations; however, the complexity of these problems has remained undetermined. Research on these problems has instead focused on determining the number of configurations for special classes of point sets~\cite{FlaNoy-DM-99,Anc-JCTA-03,KaiZie-SC-03}, bounding the number of configurations as a function of the number of points~\cite{AicHacHue-GC-07,AicOrdSan-JCTA-08,ShaWel-SICOMP-06,ShaShe-EJC-11,DumSchShe-SIDMA-13,ShaSheWel-JCTA-13,AicAlvHac-SoCG-16,SanSei-JCTA-03,AicHurNoy-CGTA-04,Sei-Comb-98,GarNoyTej-CGTA-00,AsiRot-CGTA-18},  developing exponential- or subexponential-time counting algorithms~\cite{Bes-CGTA-02,AlvSei-SOCG-13,Wett-JoCG-17,AlvBriCur-DCG-15,MarMil-SoCG-16,BroKetPoc-SICOMP-06}, or finding faster approximations~\cite{AlvBriRay-CGTA-15,KarLinSle-TCS-18}.

In this paper we bring the two worlds of \sharpp-completeness and counting non-crossing configurations together by proving the following theorem:

\begin{theorem}
\label{thm:main}
It is \sharpp-complete to count the number of triangulations of a given polygon.
\end{theorem}

The polygons constructed in our hardness proof have vertices with integer coordinates of polynomial magnitude. Necessarily, they have holes, as it is straightforward to count the triangulations of a simple polygon in polynomial time by dynamic programming~\cite{EpsSac-TOMACS-94,RaySei-EuroCG-04,DinQiaTsa-COCOON-05}.  Our proof strategy is to develop a polynomial-time counting reduction from the problem of counting independent sets in planar graphs  (here the independent sets are not necessarily maximum nor maximal), which was proved \sharpp-complete by Vadhan~\cite{Vad-SICOMP-01} (see also~\cite{XiaZhaZha-TCS-07}). We reduce counting independent sets in planar graphs to counting maximum-size non-crossing subsets of a special class of line segment arrangements, which we in turn reduce to counting triangulations.

\section{Preliminaries}

\subsection{Polygons and triangulation}

A \emph{planar straight-line graph} consists of finitely many closed line segments in the Euclidean plane, disjoint except for shared endpoints. The endpoints of these segments can be interpreted as the vertices, and the segments as the edges, of an undirected graph drawn with straight edges and no crossings in the plane. The \emph{faces} of the planar straight-line graph are the connected components of its complement (that is, maximal connected subsets of the plane that are disjoint from the segments of the graph). In any planar straight-line graph, exactly one unbounded face extends beyond the bounding box of the segments; all other faces are bounded. A segment of the graph forms a \emph{side} of a face if the interior of the segment intersects the topological closure of the face. As with any graph, a planar straight-line graph is $d$-regular if each of its vertices is incident to exactly $d$ line segments. If $G$ is any graph, we denote the sets of vertices or edges of $G$ by $V(G)$ or $E(G)$ respectively, and the numbers of vertices or edges by $|V(G)|$ or $|E(G)|$ respectively.

For the purposes of this paper, we define a \emph{polygon} $P$ to be a 2-regular planar straight-line graph in which there is a bounded face $\phi$ whose sides are all the segments of $P$.
If $\phi$ exists, it is uniquely determined from $P$. We call $\phi$ the \emph{interior} of $P$.
The connected components of the graph are necessarily simple cycles of line segments, exactly one of which separates $\phi$ from the unbounded face. If there is more than one connected component of the graph, we call the other components \emph{holes}, and if there are no holes we call $P$ a \emph{simple polygon}. We denote the number of vertices (or edges) of the polygon by $|P|$.

We define a \emph{triangulation} of a polygon $P$ to be a planar straight-line graph consisting of the edges of $P$ and added segments interior to $P$, all of whose vertices are vertices of $P$, 
partitioning the interior of $P$ into three-sided faces.
As is well known, every polygon has a triangulation. A triangulation of a polygon $P$ can be found in time $O(|P|\log |P|)$, for instance by constrained Delaunay triangulation, and this can be improved to $O(|P|\log h)$ for polygons with $h$ holes~\cite{BerEpp-CEG-95}.
The known exponential or sub-exponential algorithms for counting triangulations of point sets~\cite{AlvSei-SOCG-13,MarMil-SoCG-16}
can be adapted to count triangulations of polygons in the same time bounds.

\subsection{Counting complexity}
\label{sec:sharpp}

The complexity class \sharpp{} and the notion of \sharpp-completeness were introduced by Valiant~\cite{Val-TCS-79}.
\sharpp{} is defined as the class of functional algorithmic problems for which the desired output counts the accepting paths of some nondeterministic polynomial-time Turing machine.

\begin{lemma}
\label{lem:can-count}
Computing the number of triangulations of a polygon is in \sharpp.
\end{lemma}

\begin{proof}
The output is the number of accepting paths of a nondeterministic polynomial-time Turing machine that guesses the set of edges in the triangulation, and verifies that these edges form a triangulation of the input.
\end{proof}

\sharpp-hardness and \sharpp-completeness are defined using reductions, polynomial-time transformations from one problem $X$ (typically already known to be hard) to another problem $Y$ that we wish to prove hard. Three types of reduction are in common use for this purpose:
\begin{itemize}
\item \emph{Turing reductions} consist of an algorithm for solving  problem $X$ in polynomial time given access to an oracle for solving problem $Y$.
\item\emph{Polynomial-time counting reductions} consist of two polynomial-time transformations: a transformation $\sigma$ that transforms inputs to $X$ into inputs to $Y$, and a second transformation $\tau$ that transforms outputs of $Y$ back to outputs of $X$.  The reduction is valid if, for every input $\chi$ to problem $X$, $\tau(Y(\sigma(\chi)))=X(\chi)$.
\item \emph{Parsimonious reductions} consist of a polynomial-time transformation from inputs of $X$ to inputs of $Y$ that preserve the exact solution value.
 \end{itemize}
 
A problem $Y$ is defined to be \sharpp-hard for a given class of reductions if every problem $X$ in \sharpp{} has a reduction to $Y$. $Y$ is \sharpp-complete if, in addition, $Y$ is itself in~\sharpp. Composing two reductions of the same type produces another reduction, so we will generally prove \sharpp-hardness or \sharpp-completeness by finding a single reduction from a known-hard problem and composing it with the reductions from everything else in \sharpp{} to that known-hard problem.

The reductions that we construct in this work will be polynomial-time counting reductions.
However, we rely on earlier work on \sharpp-completeness of graph problems that uses the weaker notion of Turing reductions. Therefore, we will prove that our geometric problems are \sharpp-complete under Turing reductions. If the graph-theoretic results are strengthened to use counting reductions (and in particular if either counting maximum independent sets in regular planar graphs or counting independent sets in planar graphs is \sharpp-complete under counting reductions) then the same strengthening will apply as well to counting triangulations. For the remainder of this paper, however, whenever we refer to \sharpp-hardness or \sharpp-completeness, it will be under Turing reductions.

\subsection{Counting independent sets}
\label{sec:count-independent}

Vadhan~\cite{Vad-SICOMP-01} proved that it is \sharpp-complete to count independent sets in planar graphs. Although not necessary for our results, it will simplify our exposition to use a stronger form of this result by 
Xia, Zhang, and Zhao~\cite{XiaZhaZha-TCS-07}.
They proved that it is also \sharpp-complete to count
the vertex covers in a connected 3-regular bipartite planar graph (subsets of vertices that touch all edges).
A set of vertices is a vertex cover if and only if its complement is an independent set, so it immediately follows that it is also \sharpp-complete to count independent sets in 3-regular planar graphs.

\section{Red--blue arrangements}

In this section we define and prove hard a counting problem that will be an intermediate step in our hardness proof for triangulations. It involves counting maximum-size non-crossing subsets of certain special line segment arrangements. Counting maximum-size non-crossing subsets of arbitrary line segment arrangements can easily be shown to be hard: It follows from the hardness of counting maximum independent sets in planar graphs~\cite{Vad-SICOMP-01} and from the proof of Scheinerman's conjecture that every planar graph can be represented as an intersection graph of line segments~\cite{Sch-PhD-84,ChaGon-STOC-09}. So the significance of the reduction that we describe in this section is that it provides arrangements with a highly constrained form, a form that will be useful in our eventual reduction to counting triangulations.

\begin{definition}
We define a \emph{red--blue arrangement} to be a collection of finitely many line segments in the plane (specified by the Cartesian coordinates of their endpoints) with the following properties:
\begin{itemize}
\item Each line segment is assigned a color, either red or blue.
\item Each intersection point of two segments is a proper crossing point of exactly two segments.
\item Each blue segment is crossed by exactly two other segments, both red.
\item Each red segment is crossed by exactly three other segments, in the order blue--red--blue.
\item The union of the segments forms a connected subset of the plane.
\end{itemize}
\end{definition}

See the right side of \autoref{fig:redblue} for an example.
We will be interested in the maximum-size non-crossing subsets of such an arrangement: sets of as many segments as possible, no two of which cross.

In a red--blue arrangement, consider the graph whose vertices represent line segments and whose edges represent bichromatic crossings (crossings between a red segment and a blue segment).
Then this graph is a disjoint union of cycles, each of which has even length with vertices alternating between red and blue.
We call the cycles of this graph \emph{alternating cycles} of the arrangement.

\begin{lemma}
\label{lem:rb-max}
Every red--blue arrangement has equal numbers of red and blue segments. If there are $r$ red segments (and therefore also $r$ blue segments), the maximum-size non-crossing subsets of the arrangement all have exactly $r$ segments. Within each alternating cycle of the arrangement, a maximum-size non-crossing subset must use a monochromatic subset of the cycle (either all the red segments of the cycle, or all the blue segments of the cycle).
\end{lemma}

\begin{proof}
The equal numbers of red and blue segments in the whole arrangement follow from the decomposition of the arrangement into alternating cycles and the equal numbers within each cycle.
Within a single cycle, there are exactly two maximum-size non-crossing subsets, the subsets of red and of blue segments, each of which uses exactly half of the segments of the cycle. Therefore, no non-crossing subset of the whole arrangement can include more than $r$ segments (half of the total), and a non-crossing subset that uses exactly $r$ segments must be monochromatic within each alternating cycle. There exists at least one non-crossing subset of exactly $r$ segments, namely the set of blue segments.
\end{proof}

Given a connected 3-regular planar graph $G$, we will construct a red--blue arrangement $A_G$ from it. In overview, our construction begins by finding a straight-line drawing of $G$.
We then replace each vertex $v$ of $G$ by a twelve-segment alternating cycle (\autoref{fig:redblue}, lower left).
Each edge of $G$ incident to $v$ will have two red segments on either side of it, within the two faces bounded by that edge, and these segments form the six red segments of the alternating cycle.
Three blue segments, drawn near $v$ within the three faces incident to $v$, connect pairs of red segments within each face. Another three blue segments cross the three edges incident to $v$ (which are not themselves part of the arrangement) and connect the two red segments on either side of the edge.

\begin{figure}[t]
\includegraphics[width=\textwidth]{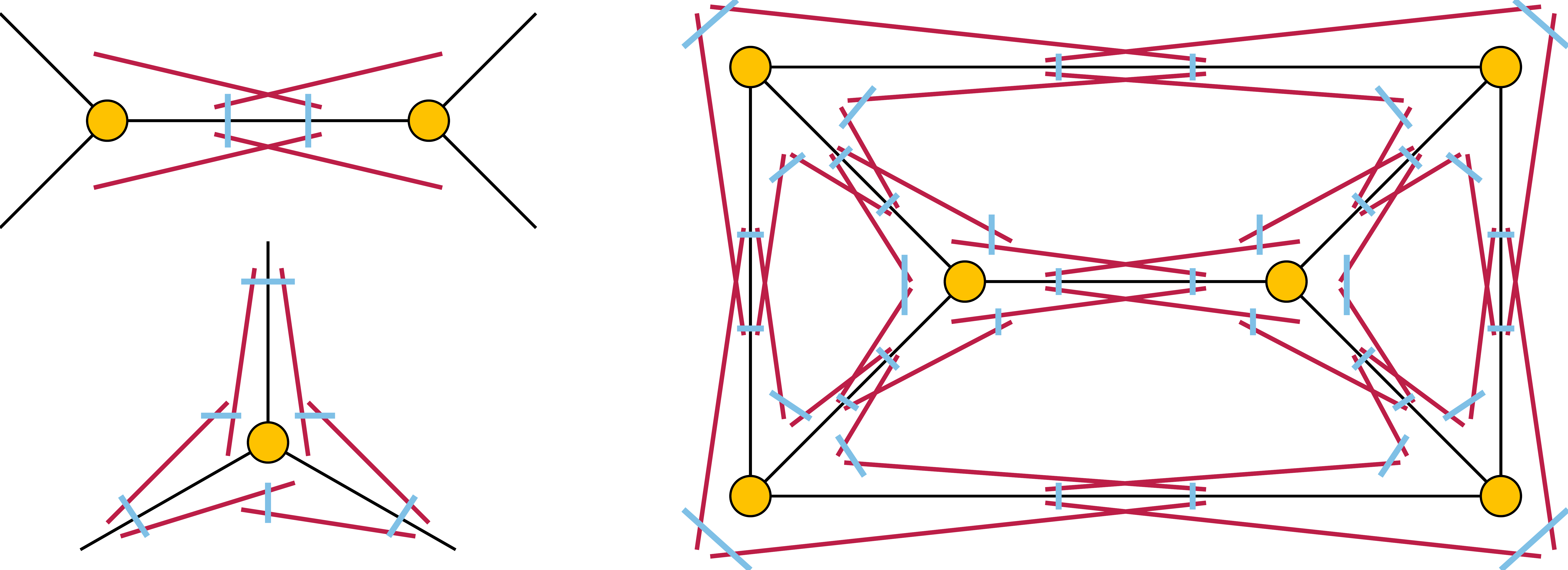}
\caption{Transformation from a 3-regular planar graph $G$ to a red--blue arrangement. Left: gadgets for transforming edges and vertices of $G$ into red and blue segments that will form parts of the arrangement. Right: an example of the complete red--blue arrangement for the graph of a triangular prism.}
\label{fig:redblue}
\end{figure}

In this way, each edge $uv$ of $G$ will have four red segments near it (two on each side from each of its two endpoints) and two blue segments crossing it (one from each endpoint). We arrange these segments so that, on each side of $uv$, the two red segments cross, and there are no other crossings except those in the alternating cycles (\autoref{fig:redblue}, upper left). An example of the whole red--blue arrangement produced by these rules is shown in \autoref{fig:redblue}, right.

In order to apply this drawing method as part of a counting reduction, it needs to be constructable in polynomial time, using coordinates that can be represented by binary numbers with a polynomial number of bits of precision. We will show more strongly that the coordinates themselves can be chosen to have polynomial magnitude (needing only a logarithmic number of bits of precision).
Thus, we now describe in more detail the steps of our construction.

The first of these steps is to find a planar straight line drawing of the given planar graph $G$. We use known methods to find a planar straight line drawing of $G$ within an integer grid of size $|V(G)|\times |V(G)|$, in linear time~\cite{Sch-SODA-90,FraPacPol-Comb-90}.

\begin{figure}[t]
\includegraphics[width=\textwidth]{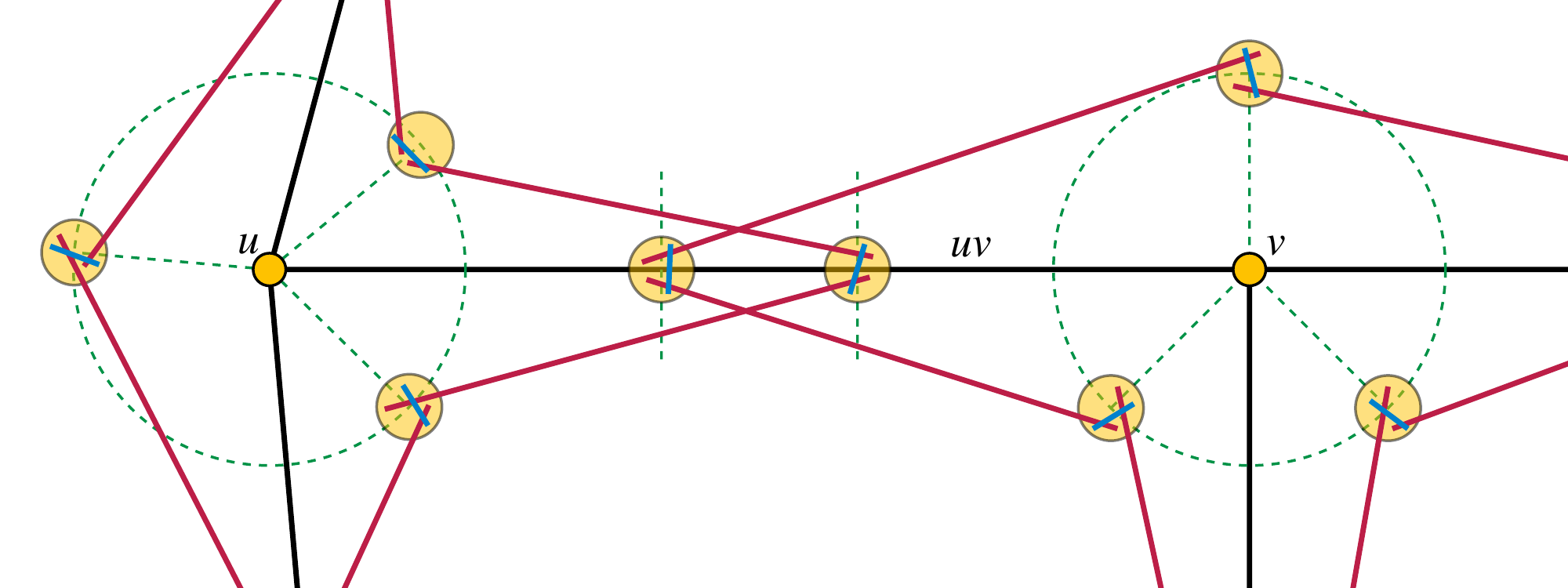}
\caption{Part of a graph ~$G$ including two vertices $u$ and $v$ (small yellow circles), the edge $uv$ between them, parts of other adjacent edges (black), the associated guide disks (large yellow circles), and a placement of red and blue segment endpoints within the guide disks.}
\label{fig:guidedisks}
\end{figure}

After constructing the drawing of $G$, the next step of our construction is to place ``guide disks'' in the drawing, two along each edge and three near each vertex (\autoref{fig:guidedisks}). These disks will each contain two endpoints of disjoint red segments in our eventual red-blue segment arrangement, and both endpoints of a blue segment that crosses these two red segments. We center the two guide disks on each edge at points $2/5$ and $3/5$ of the way between the two endpoints of the edge.
At each vertex $v$, three edges of $G$ are incident, forming the boundaries of three wedges between them. We place the three guide disks on the angle bisectors of these wedges, centered at a distance from $v$ equal to $1/5$ of the smallest nonzero distance from $v$ to any other edge or vertex of $G$. We choose all of these guide disks to have equal radii $\rho$ (depending on $|V(G)|$). We set $\rho$ to be small enough to meet the following constraints:
\begin{itemize}
\item No two guide disks intersect each other.
\item A guide disk on an edge of $G$ does not intersect any other edge or vertex of $G$.
\item A guide disk near a vertex of $G$ does not intersect any edge or vertex of $G$.
\item For every edge $uv$ of $G$, there does not exist a line segment in the plane that intersects the two guide disks on $uv$ and a third guide disk in one of the four wedges bounded by $uv$.
\end{itemize}

\begin{lemma}
There exists a radius $\rho=\Theta(1/|V(G)|^3)$ that satisfies all of the constraints above.
\end{lemma}

\begin{proof}
Each two vertices in $G$ are at least at unit distance, because $G$ is on a grid.
The distance between each edge and non-incident vertex of $G$ is $\Omega(1/|V(G)|)$,
because the convex hull of the edge and vertex has area $\Omega(1)$ (by Pick's formula) and diameter $O(|V(G)|)$.
Each two incident edges form an angle of $\Omega(1/|V(G)|^2)$, because they form a triangle of area $\Omega(1)$ and diameter $O(|V(G)|)$ and because the area of a triangle is proportional to the product of its squared diameter with its sharpest angle.
So, for each guide disk associated with a vertex $v$, the distance of the disk center from $v$ is $\Omega(1/|V(G)|)$, and it will avoid intersecting the two nearby edges forming an angle of $\Omega(1/|V(G)|^2)$) as long as $\rho$ is a sufficiently small constant multiple of $1/|V(G)|^3$.
The same bound suffices to achieve all of the instances of the first three constraints. It remains to consider the final constraint.

This constraint may equivalently be stated as requiring the height of the triangle formed by certain triples of guide disk centers, above the longest edge of the triangle to be at least $2\rho$. In these triangles, for a triple of disk centers associated with vertex $u$ and edge $uv$, the center near $u$ has height $\Omega(1/|V(G)|^3)$ above line $uv$, which contains one of the sides of the triangle.
But in this triangle, all three edges have lengths within a constant factor of each other, so all heights are also within a constant factor of each other and are all $\Omega(1/|V(G)|^3)$. So when $\rho$ is a sufficiently small constant multiple of $1/|V(G)|^3$, the last constraint is also satisfied.
\end{proof}

We must next show how to choose the endpoints of the red and blue segments within each guide disk. To do so, we expand our drawing of $G$ (and its guide disks) by a factor of $\Theta(|V(G)|^3)$, enough to ensure that each guide disk contains an $8\times 8$ grid of points. We observe that (because of the small size of the guide disks relative to their distance apart) the angles of each red segment are known to within an additive error of $O(1/|V(G)|^3)$ regardless of where in their guide disks their endpoints are placed. Additionally, the angle between any two red segments that leave the same guide disk is $\Omega(1/|V(G)|^2)$, again regardless of where in their guide disks the segments are placed. We choose the endpoints of segments within each guide disk independently, using the following lemma.

\begin{lemma}
\label{lem:8x8suffices}
There exists a constant $\delta$ such that, for every two given slopes $\Theta_1$ and $\Theta_2$ and every  $8\times 8$ grid of points there exist two rays and a line segment with the following properties:
\begin{itemize}
\item One ray has slope $\Theta_1$ and the other has slope $\Theta_2$.
\item Both rays have their endpoint in the grid, and the line segment has both endpoints in the grid.
\item The two rays are disjoint from each other, and both are crossed by the line segment.
\item If the angles of the rays are continuously rotated by at most $\delta$ (while not passing through the same angle as each other) the rays remain disjoint from each other and crossed by the segment.
\end{itemize}
\end{lemma}

\begin{figure}[t]
\centering\includegraphics[width=0.8\textwidth]{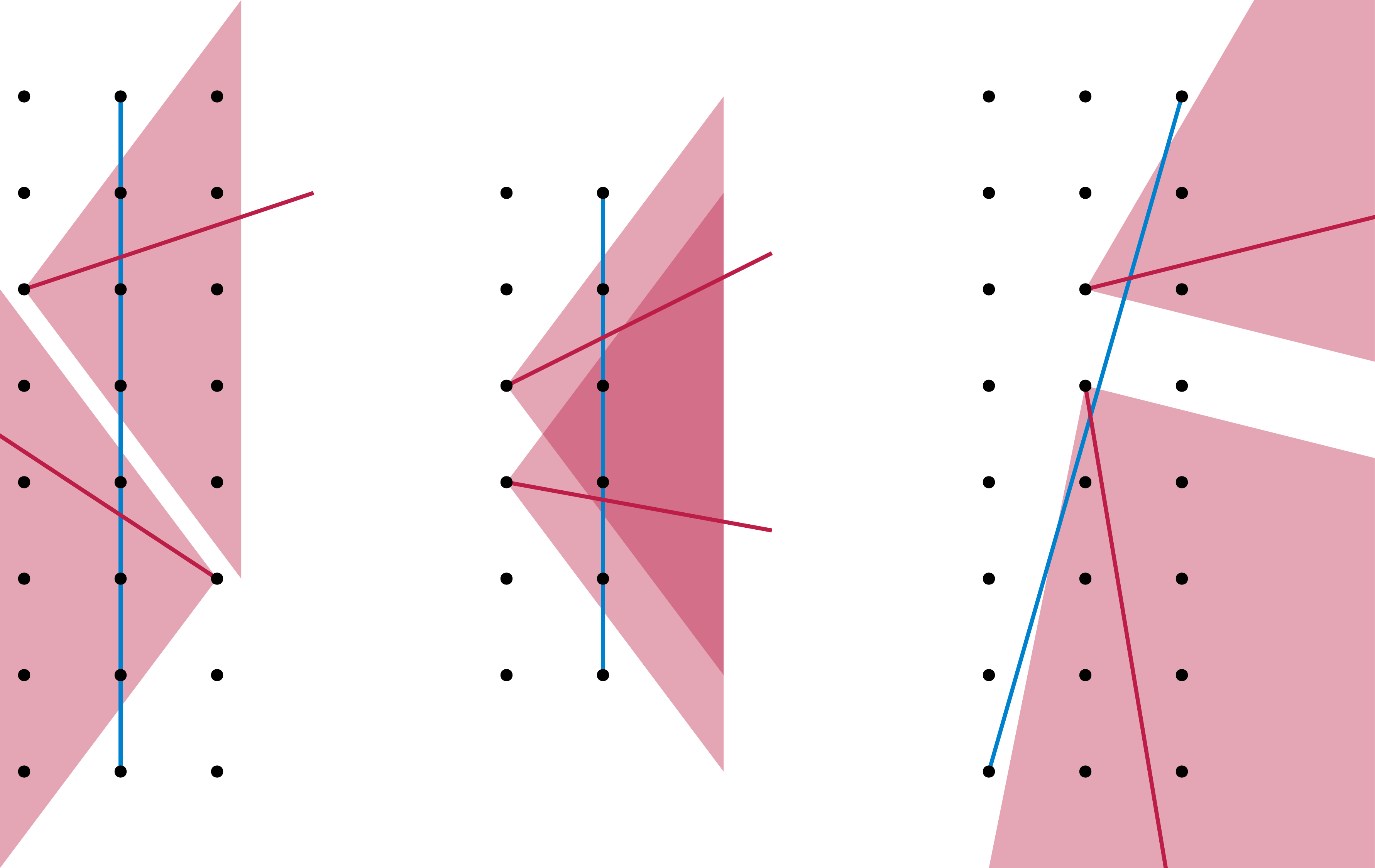}
\caption{Case analysis for \autoref{lem:8x8suffices}. For any two given slopes we can find a grid of size at most $8\times 8$, and two disjoint (red) rays with those slopes, having endpoints within the grid and crossed by a blue segment with both endpoints in the grid,
such that any sufficiently-small rotation of the red rays does not change this crossing pattern. Left: two slopes that are within $45^\circ$ of opposite directions on a coordinate axes;. Center: two slopes that are within $45^\circ$ of the same direction on a coordinate axis. Right: two slopes that are close to forming a right angle. The left and middle cases also occur rotated by $45^\circ$, but these rotations still fit within an $8\times 8$ grid.}
\label{fig:placement}
\end{figure}

\begin{proof}
We may assume without loss of generality that the given grid is axis-aligned. Some of the cases in our case analysis will instead use a grid that is rotated by $45^\circ$ from the coordinate axes; such a grid can be found as a subset of the given grid points, forming a subgrid with spacing larger by a factor of $\sqrt{2}$.

If $\Theta_1$ and $\Theta_2$ are within $45^\circ$ of opposite coordinate axes, or both within $45^\circ$ of opposite directions along a line of slope $\pm 1$, we may use either the placement of red rays and blue segment shown in \autoref{fig:placement} (left), or its rotation by a multiple of $45^\circ$ within a rotated subgrid. The grid placement shown uses a $3\times 8$ subgrid of the given $8\times 8$ grid,
while the rotated placement fits within an $8\times 8$ axis-aligned grid (the bounding box of the rotated version of the blue segment). The red wedges in the figure show regions within which the two rays may rotate while preserving the correctness of the construction; note that they span an angle greater than $45^\circ$ from the horizontal axis in the figure. We may take $\delta$ in this case to be half of the amount by which they are greater than $45^\circ$.

If $\Theta_1$ and $\Theta_2$ are both within $45^\circ$ of the same direction along a coordinate axis, or along a line of slope $\pm 1$, we may use the placement of red rays and blue segment shown in \autoref{fig:placement} (center) or its rotation by a multiple of $45^\circ$. The unrotated case fits within a $2\times 5$ grid, and the rotated case fits within a $5\times 5$ grid (the bounding box of the rotated version of the blue segment). Again, the wedges span an angle greater than $45^\circ$ and we may take $\delta$ to be half of the amount by which they are greater.

In the remaining case, if we partition the circle of directions of rays into eight arcs spanning angles of $45^\circ$ by splitting it at the directions of coordinate axes and lines of slope $\pm 1$, then $\Theta_1$ and $\Theta_2$ must belong to arcs that are at right angles to each other. In this case, we may use the placement of red rays and blue segment shown in \autoref{fig:placement} (right) or its rotation by a multiple of $90^\circ$. This case fits within a $3\times 8$ grid. Again, the red wedges cover a greater range of angles than the $45^\circ$ arcs that $\Theta_1$ and $\Theta_2$ belong to, and we may take $\delta$ to be half of the amount by which they are greater.
\end{proof}

By choosing segment endpoints in this way within each guide disk, we ensure that the red and blue segments cross each other in the pattern described earlier. This produces an alternating 6-cycle of red and blue segments around each vertex. The cycles for adjacent vertices cross each other twice.

The next lemma shows that the construction of $A_G$ from $G$ is a parsimonious reduction from independent sets to maximum non-crossing subsets, and moreover that the size of an independent set can be determined from the colors in the corresponding maximum non-crossing subset.

\begin{lemma}
\label{lem:rb-parsimony}
Let $G$ be a connected 3-regular planar graph, and let $A_G$ be constructed from $G$ as above.
Then the independent sets of $G$ correspond one-for-one with the maximum-size non-crossing subsets of $A_G$, with an independent set having size $k$ if and only if the corresponding non-crossing subset has exactly $6k$ red segments.
\end{lemma}

\begin{proof}
Given an independent set $Z$ in $G$, one can construct a non-crossing subset of $A_G$ of size $6|V(G)|$, with $6|Z|$ red segments, by choosing all the red segments in the alternating cycles around members of $Z$, and all the blue segments in the remaining alternating cycles. Because $Z$ includes no two adjacent vertices, the subsets chosen in this way will be non-crossing. $Z$ can be recovered as the set of vertices from which we used red segments of alternating cycles, so the correspondence from independent sets to non-crossing subsets constructed in this way is one-to-one.

By \autoref{lem:rb-max}, these non-crossing subsets have maximum size. By the same lemma, every maximum-size non-crossing subset consists of red segments in the alternating cycles of some vertices of $G$ and blue segments in the remaining alternating cycles. No two adjacent vertices of $G$ can have their red segments chosen as that would produce a crossing,
so there are no maximum-size non-crossing subsets other than the ones coming from independent sets of~$G$.
\end{proof}

This leads to the main result of this section:

\begin{lemma}
\label{lem:red-blue-hard}
It is \sharpp-complete to compute the number of maximum-size non-crossing subsets of a red--blue arrangement $A$.
\end{lemma}

\begin{proof}
The problem is clearly in \sharpp, as the number of these subsets equals the number of accepting paths in a nondeterministic polynomial-time Turing machine that nondeterministically chooses a subset of segments of the arrangement, verifies that no two of the chosen segments cross each other, verifies that the number of chosen segments is half of the total number of segments, and accepts only if both of these things are true.
By \autoref{lem:rb-parsimony}, the construction of the red--blue arrangement $A_G$ from a 3-regular planar graph $G$ is a parsimonious reduction from the known \sharpp-complete problem of counting independent sets in 3-regular planar graphs.
\end{proof}

The arrangement constructed by this reduction has integer coefficients of magnitude $O(|V(G)|^4)$. Here, one factor of $|V(G)|$ comes from the $|V(G)|\times |V(G)|$ grid on which we drew the planar graph $G$, and the remaining factor of $O(|V(G)|^3)$ comes from the expansion needed to ensure that each guide disk contains an $8\times 8$ grid of integer points.

\section{Reduction to counting triangulations}

In this section we describe our reduction from line segment arrangements to polygons. The rough idea is to thicken each segment of the given red--blue arrangement to a rectangle with rounded ends, and form the union of these rectangles. In the part of the polygon formed from each segment, many more triangulations will use diagonals running end-to-end along the segment than triangulations that do not, causing most triangulations to correspond to sets of diagonals from a maximum-size non-crossing subset of the arrangement. With a careful choice of the shapes of the thickened segments in this construction, we can recover the number of maximum-size non-crossing subsets of the arrangement from the number of triangulations.

\subsection{From segments to polygons}

\begin{figure}[t]
\centering\includegraphics[width=\textwidth]{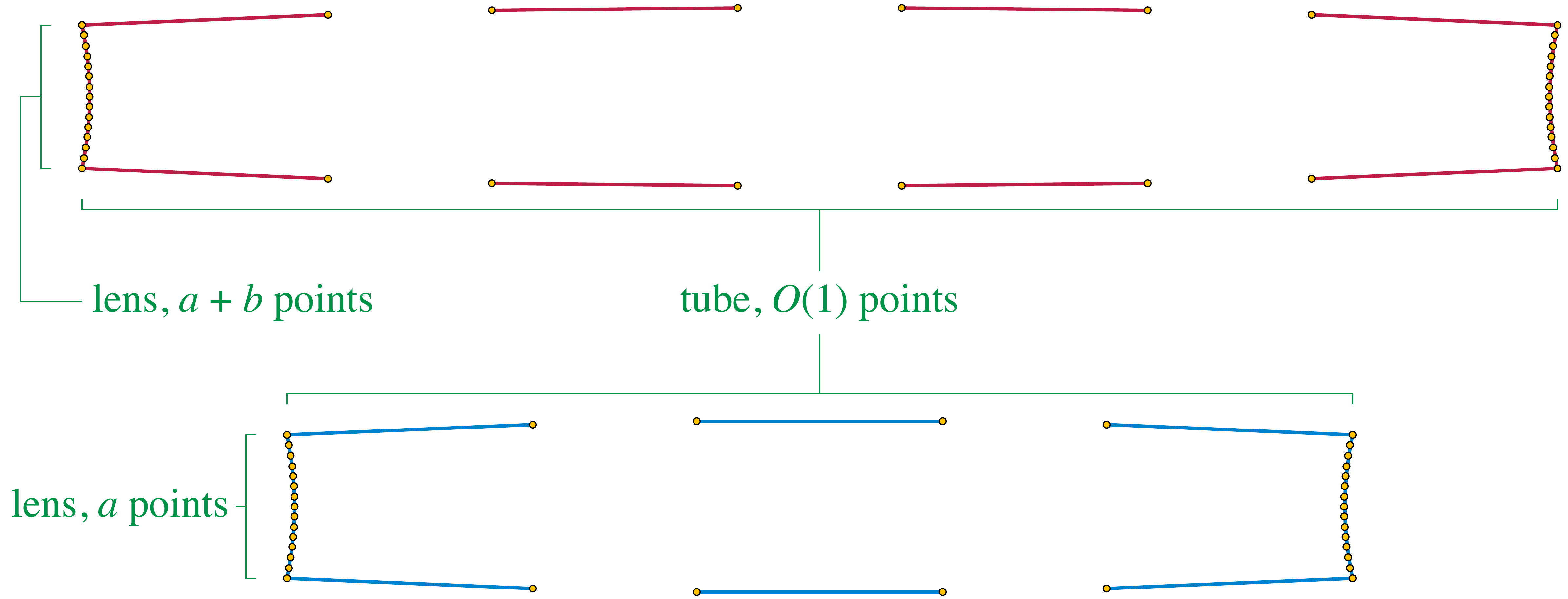}
\caption{Polygonal gadgets for replacing the red segments (top) and blue segments (bottom) of a red--blue arrangement}
\label{fig:telescope}
\end{figure}

The gadgets that we use to replace the red and blue segments of a red--blue arrangement are shown (not to scale) in \autoref{fig:telescope}. The top part of the figure shows the gadget used to replace a single red segment, and the bottom part of the figure shows the gadget used to replace a single blue segment. These gadgets resemble the cross-section of a telescope, and we have named their parts accordingly. Each gadget consists of $O(1)$ vertices spaced along two convex curves, which we call the tube. These two curves bend slightly outwards from each other, but both remain close to the center line of the gadget. Consecutive pairs of points are connected to each other along the tube, leaving either three gaps (for red segments) or two gaps (for blue segments) where other segments of the arrangement cross the given segment. The two vertices on each side of each gap are shared with the gadget for the crossing segment.

At the ends of the tubes, the two convex curves are connected to each other by two concave curves (the lenses of the gadget), each containing larger numbers of vertices that are connected consecutively to each other without gaps. Each blue lens has some number $a$ of vertices,
and each red lens has some larger number $a+b$ of vertices, where $a$ and $b$ are both positive integers to be determined later as a function of the number of segments in the arrangement.
The key geometric properties of these gadgets are:

\begin{description}
\item[P1.] Each gadget forms a collection of polygonal chains whose internal vertices (the vertices that are not the end-point of any of the chains) are not part of any other gadget and whose endpoints are part of exactly one other gadget. The union of all the gadgets of the arrangement forms a single polygon (with holes).
\item[P2.] Each vertex of the lens at one end of the gadget is visible within the polygon to each vertex of the lens at the other end of the gadget, and to each vertex of the tube of the gadget, but not to any vertices of its own lens that it is non-adjacent to.
\item[P3.] For each pair of points that are visible to each other there is a gadget containing both of them. (This is not an automatic property of the gadgets as drawn, but can be achieved by making the gadgets sufficiently narrow relative to the spacing of points along their tubes, depending on the angles at which different segments cross each other.)
\end{description}

\begin{figure}[t]
\centering\includegraphics[width=0.8\textwidth]{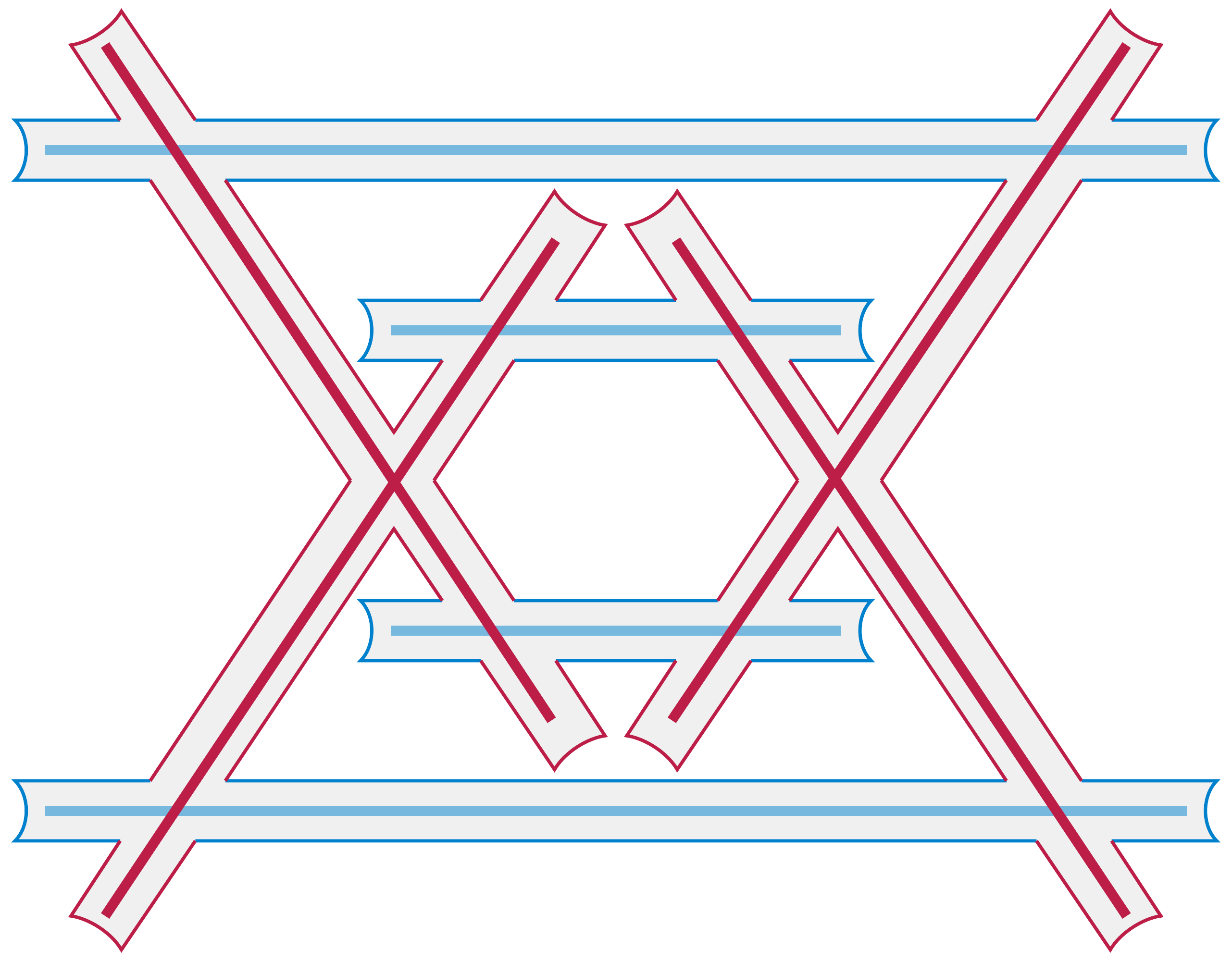}
\caption{A red--blue arrangement $A$ (thick colored segments) and the polygon $P_A$ (shaded region with thin outlines) obtained from it by replacing its segments by gadgets (schematic view, not to scale). $P_A$ has three holes (unshaded regions entirely surrounded by $P_A$): a non-convex region below the upper blue gadget, a central convex hexagon, and another non-convex region above the lower blue gadget.}
\label{fig:arrpoly}
\end{figure}

\begin{definition}
\label{def:polygons}
We denote by $P_A$ any polygon obtained by replacing each segment of a red--blue arrangement $A$ by a gadget, obeying properties P1, P2, and P3.
For any gadget $S$, we let $\Pi_S$ be the polygon formed by intersecting $P_A$ with the convex hull of $S$. $\Pi_S$ connects the endpoints of the polygonal chains of $S$ into a single polygon in the obvious way.
\end{definition}

\autoref{fig:arrpoly} gives an example of a red--blue arrangement $A$ and its corresponding polygon $P_A$, drawn approximately and schematically (not to scale).
To analyze the magnitudes of the coordinates needed to carry out this construction,
let $\mu$ denote the maximum magnitude of the integer coordinates in the given arrangement $A$,
let $\nu$ denote the sharpest angle at which two segments of $A$ cross each other,
and let $\xi$ denote the minimum distance between two crossing points or endpoints of the same segment of $A$. For instance, for the arrangement generated by the reduction from planar graphs to red-blue arrangements of \autoref{lem:red-blue-hard}, we have $\mu=O(|V(G)|^4)$, as discussed at the end of \autoref{sec:count-independent}.
We also have $\nu=\Omega(1/|V(G)|^4)$: the crossing angles between red and blue segments are described by the case analysis of \autoref{lem:8x8suffices}, and are all $\Omega(1)$, so the sharpest possible crossings are those between two red segments. Each red segment has one endpoint within a guide disk of radius $O(1)$ centered on a point in an edge of $G$, and forms an angle with that edge large enough to be disjoint from the other guide disk; since the length of the edge is $O(|V(G)|^4)$, the angle between the red segment and the edge must be $\Omega(1/|V(G)|^4)$. Any two crossing red segments have angles of opposite signs to the edge they both correspond to, of at least this magnitude, so their crossing angle is $\Omega(1/|V(G)|^4)$. And we have $\xi=\Omega(1)$ because all nearby crossing pairs occur within one of the cases of \autoref{lem:8x8suffices}, which all have constant separation for their two crossings.

We will scale $A$ by a large integer factor so that the integer grid forms a fine subgrid overlain on the expanded copy of $A$, allowing us to construct $P_A$ using vertices with integer coordinates. We now examine more carefully how big this scaling factor should be, as a function of $\mu$, $\nu$, and $\xi$, in order to construct $P_A$.

In order to achieve property P3, we set the width of each gadget in the scaled grid to be $O(\xi\nu)$. To be able to maintain this width, over a gadget whose length could be as much as $\mu$, we need the slopes of the tube edges to be nearly parallel, but varying by each other by angles of magnitude $\Theta(\xi\nu/\mu)$. And to be able to construct segments of these angles using vertices of integer coordinates, to use as the tube edges of a gadget, it suffices to choose the scaling factor of the grid to be $O(\xi^2\nu/\mu)$

We must also be able to construct the lenses of each gadget. It is possible to form a convex chain of $a+b$ vertices within a square grid whose side length is $\Theta((a+b)^{3/2})$~\cite{Jar-MZ-26,Epp-18}. Because our gadgets are very long and narrow, the requirement that the two opposite lenses be completely visible to each other does not significantly affect this bound. In order to fit a square grid large enough to contain such a convex chain within the width of a single gadget, it suffices to choose the scaling factor of the grid to be $O((a+b)^{-3/2}\xi\nu)$.

Putting these two bounds for the scaling factor needed for the tubes and lenses of our gadgets,
we can perform our construction with a scaling factor of
\[
\Theta\left(\max\left(\frac{\mu}{\xi^2\nu},\frac{(a+b)^{3/2}}{\xi\nu}\right)\right).
\]
Plugging in the values $\mu=O(|V(G)|^4)$, $\nu=\Omega(1/|V(G)|^4)$, and $\xi=\Omega(1)$ coming from \autoref{lem:red-blue-hard} gives us a scale factor of
\[
\Theta\left(\max\left(|V(G)|^8,(a+b)^{3/2}|V(G)|^4\right)\right).
\]
When $a$ and $b$ are bounded by a polynomial of $|V(G)|$ (as they will be),
this allows our construction to be performed with integer coordinates of polynomial magnitude.

\subsection{Triangulation within a gadget}

By looking at the way a triangulation of $P_A$ behaves within each gadget, we can recover a non-crossing subset of segments of $A$ from the triangulation, as the following definitions and lemmas show.

\begin{figure}[t]
\centering\includegraphics[width=\textwidth]{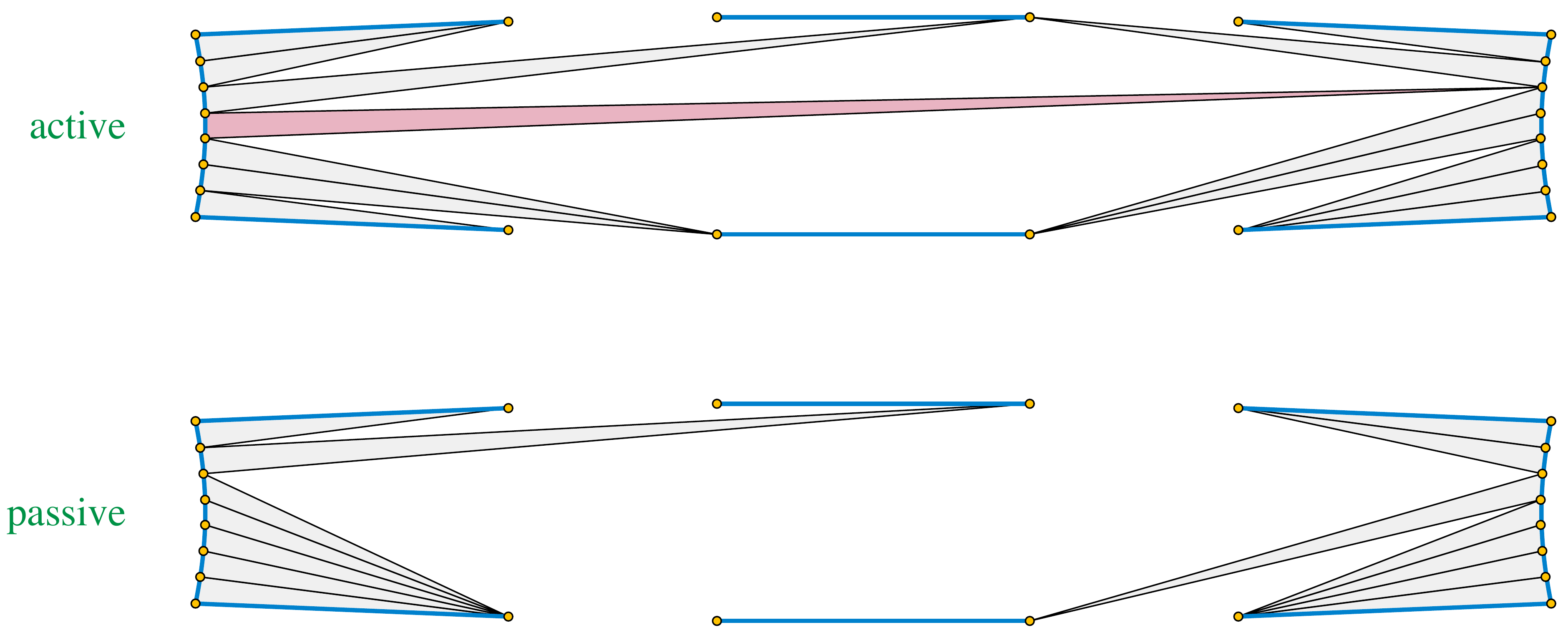}
\caption{The local part of a triangulation in the gadget of a segment that is active (top) or passive (bottom). A triangle of the top local part that stretches from lens to lens, making the segment active, is highlighted by red shading.}
\label{fig:active}
\end{figure}

\begin{definition}
For a triangulation $T$ of a polygon $P_A$, and a gadget $S$ of $P_A$ corresponding to a segment $s$ of $A$, we define the \emph{local part} of $T$ in $S$ to be the set of triangles that have a side on one of the two lenses of $S$.  We say that segment $s$ of $A$ is \emph{active} in $T$ if its local part includes at least one triangle that touches both lenses of $S$, and that $s$ is \emph{passive} otherwise.
\end{definition}

See \autoref{fig:active} for an example, and note that (as in the figure) even when a segment is passive, its triangulation may include triangles that block visibility across the gadget.

\begin{lemma}
\label{lem:active-non-crossing}
The active segments of $A$ in any triangulation of $P_A$ form a non-crossing set.
\end{lemma}

\begin{proof}
If two segments cross, any lens-to-lens triangle within one segment crosses any lens-to-lens triangle of the other. But in a triangulation, no two triangles can cross each other.
\end{proof}

\begin{lemma}
\label{lem:active-triangulation}
Let $T$ be a triangulation of $P_A$, let $s$ be an active segment of  $A$ in $T$,  let $S$ be the corresponding gadget, bounded by polygon $\Pi_S$. Then the intersection of $T$ and $\Pi_S$ is a triangulation of~$\Pi_S$.
\end{lemma}

\begin{proof}
By the construction of $P_A$, every edge of $T$ must belong to a single gadget (as there are no other visibilities between pairs of vertices in $P_A$). By the same argument as in \autoref{lem:active-non-crossing}, no segment with its endpoints in a single other gadget than $S$ can cross $\Pi_S$. Therefore, all edges of $T$ that include an interior point of $\Pi_S$ must connect two vertices of $S$, and lie within $\Pi_S$. It follows that all triangles of $T$ that include an interior point of $\Pi_S$ must also lie within $\Pi_S$. These triangles cover $\Pi_S$ and lie entirely within $\Pi_S$, so they must form a triangulation of~$\Pi_S$.
\end{proof}

\begin{corollary}
\label{cor:passive-triangulation}
Let $T$ be a triangulation of $P_A$, and let $L$ be a lens in polygonal chain $C_L$ of a gadget $S$. Suppose that some active segment of $A$ in $T$ separates $C_L$ from the rest of $S$. Let $Q_L$ be the polygon formed from $C_L$ by adding one more edge connecting the two endpoints of $C_L$. Then $T$ intersects $Q_L$ in a triangulation of~$Q_L$.
\end{corollary}

\begin{proof}
By \autoref{lem:active-triangulation}, $T$ contains the edge connecting the endpoints of $C$. This edge separates $Q_L$ from the rest of the polygon, so each triangle of $T$ must lie either entirely within or entirely outside $Q_L$. But the triangles within $Q_L$ cover $Q_L$ (as the triangles of $T$ cover all of $P_A$) so they form a triangulation of~$Q_L$.
\end{proof}

These claims already allow us to produce a precise count of the triangulations whose active segments form a maximum-size non-crossing subset, as a function of the number of red segments in the subset, which we will do in the next section. However, we also need to bound the number of triangulations with smaller sets of active segments, and show that (with an appropriate choice of the parameters $a$ and $b$) they form a negligible number of triangulations compared to the ones coming from maximum-size non-crossing subsets.

\subsection{Counting within a gadget}

It is straightforward to count triangulations of the polygon formed by a single lens:

\begin{figure}[t]
\centering\includegraphics[width=0.8\textwidth]{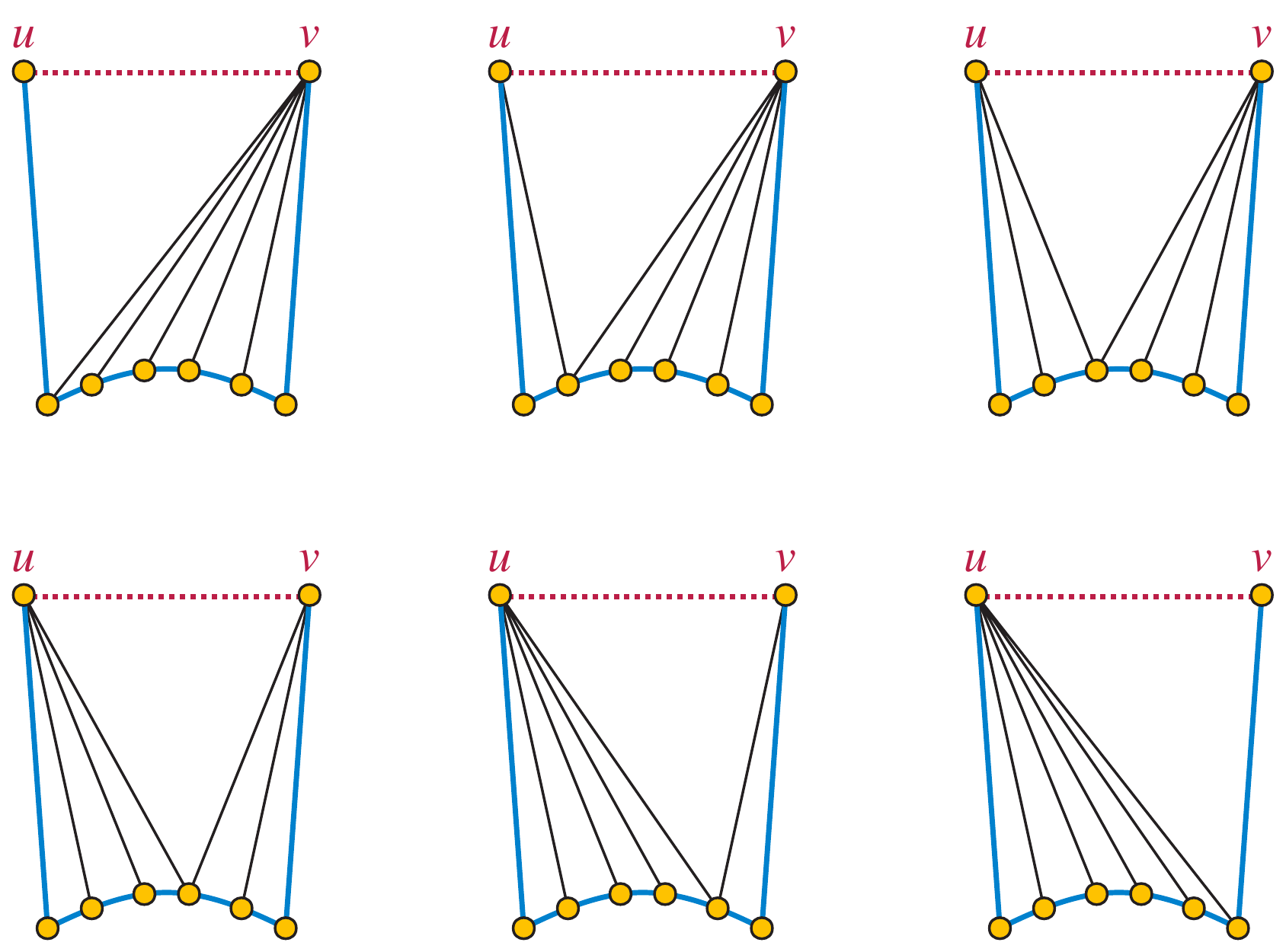}
\caption{Illustration for \autoref{lem:count-lenstri}: Closing off a lens $L$ by connecting the endpoints of its path (red dashed edge $uv$) forms a polygon $Q_L$ with as many distinct triangulations as lens vertices.}
\label{fig:lenstri}
\end{figure}

\begin{lemma}
\label{lem:count-lenstri}
Let $Q_L$ be the polygon formed from a lens by closing its path off by a single edge~$uv$, as in \autoref{cor:passive-triangulation}. Then the number of distinct triangulations of $Q_L$ equals the number of vertices of the lens ($a$ for a blue lens, $a+b$ for a red lens).
\end{lemma}

\begin{proof}
Any triangulation of $Q_L$ is completely determined by the choice of apex of the triangle that has $uv$ as one of its sides. All the remaining edges of the triangulation must each connect one vertex of the lens to the remaining visible endpoint of $uv$, the only vertex visible to that lens vertex (\autoref{fig:lenstri}). The number of choices for the apex equals the number of lens vertices.
\end{proof}

More generally, for any passive segment, we have:

\begin{lemma}
\label{lem:passive-is-poly}
Let $S$ be the gadget of segment $s$. Then the number of combinatorially distinct ways to choose the local part in $S$ of a triangulation of $P_A$ for which $s$ is passive is polynomial in the number of lens vertices of $S$.
\end{lemma}

\begin{proof}
Recall that the triangles of the local part each have one lens edge as one of their sides. Within each lens, the only choice is where to place the apex of each of these triangles. There are $O(1)$ choices of apex vertex, and within each lens the edges whose triangles have a given apex must form a contiguous subsequence. There are only polynomially many ways of partitioning the sequences of lens edges into a constant number of contiguous subsequences.
\end{proof}

\begin{lemma}
\label{lem:nonlocal}
Let $T$ be a triangulation of $P_A$ and let $S$ be a gadget of $P_A$. Then only $O(1)$ vertices of $S$ can participate in triangles of $T$ which are not part of the local part of $T$ in $S$.
\end{lemma}

\begin{proof}
$S$ has $O(1)$ vertices outside of its two lenses, so it has $O(1)$ non-lens vertices that can participate in non-local triangles. A lens vertex can participate in a non-local triangle in one of two ways: either the triangle has one lens vertex and two non-lens vertices, or it has two vertices from opposite lenses and one non-lens vertex. There are $O(1)$ edges of $T$ between non-lens vertices of $S$,  each of which is part of two triangles in $T$, so there are $O(1)$ triangles in $S\cap T$ with only one lens vertex. Each non-lens vertex of $S$ can participate in only one triangle whose other two vertices are on opposite lenses, so again there are $O(1)$ triangles of this type. Therefore, there are $O(1)$ lens vertices in non-local triangles.
\end{proof}

It is sufficiently messy to count the triangulations of active gadgets that we provide here only approximate bounds. However, the exact number of triangulations can be found in polynomial time using the algorithm for counting triangulations of simple polygons~\cite{EpsSac-TOMACS-94,RaySei-EuroCG-04,DinQiaTsa-COCOON-05}.

\begin{figure}[t]
\centering\includegraphics[width=0.8\textwidth]{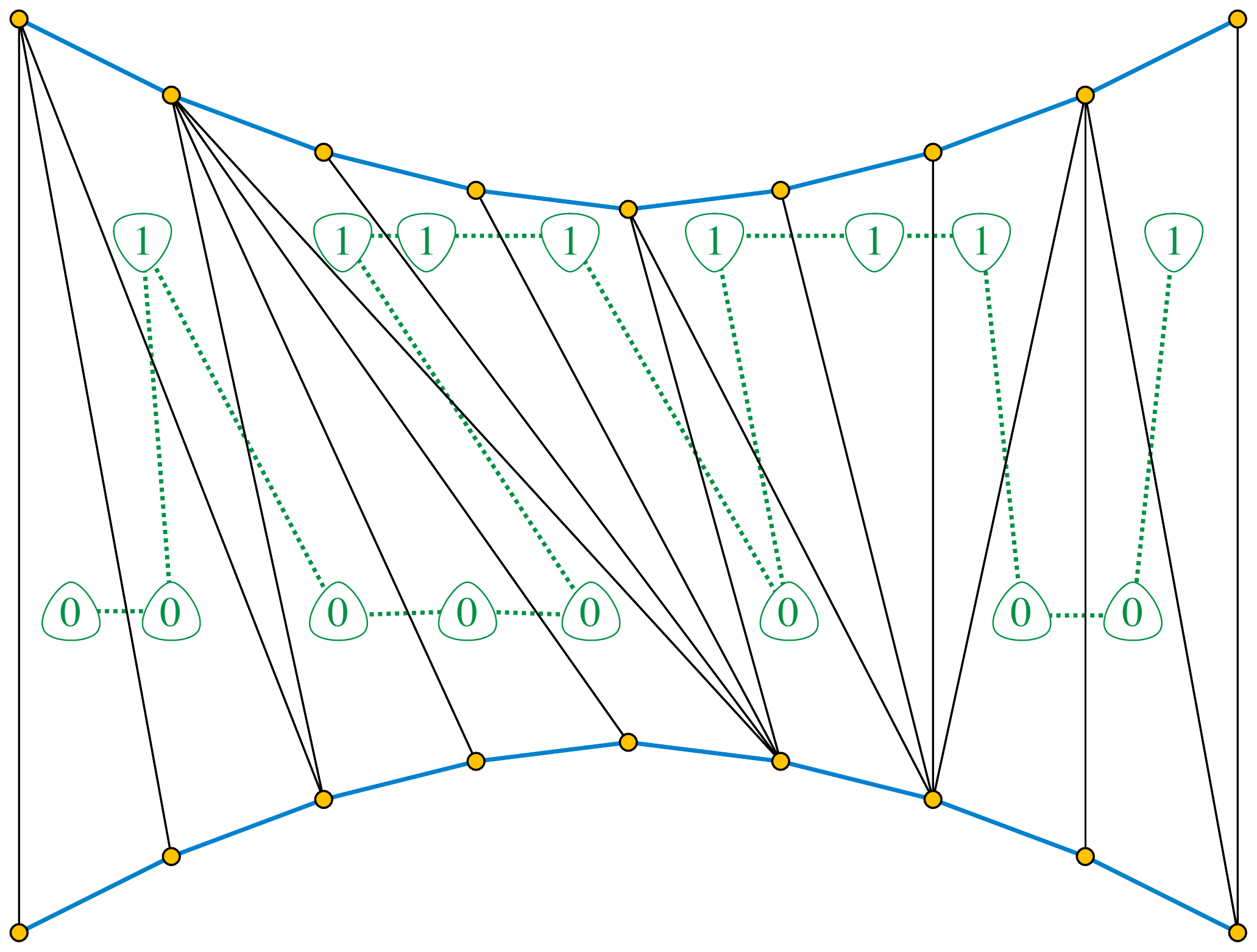}
\caption{Encoding a triangulation of the convex hull of two lenses by a sequence of bits}
\label{fig:2lenstri}
\end{figure}

\begin{lemma}
\label{lem:count-active}
Let $\Pi_S$ be the polygon bounding gadget $S$ (from \autoref{def:polygons}). Let $\ell$ be the number of lens vertices of $\Pi_S$ ($2a$ for a blue gadget, $2a+2b$ for a red gadget). Then the number of triangulations of $\Pi_S$ is $\Theta(2^\ell/\sqrt\ell)$.
\end{lemma}

\begin{proof}
One can obtain $\Omega(2^\ell/\sqrt\ell)$ triangulations by choosing a maximal set of non-crossing lens-to-lens diagonals in $\Pi_S$, which form the edges of a triangulation of the convex hull of the two lenses, and then choosing arbitrarily a triangulation of the remaining polygons of $O(1)$ vertices on either side of the convex hull of the two lenses.
The triangles within any triangulation of the convex hull of the two lenses have a path as their dual graph,
and if they are ordered along the path then the triangulation itself is determined by a sequence of bits (one bit per triangle) that denote whether each triangle includes an edge of one lens or of the other lens \cite{HurNoyUrr-DCG-99} (\autoref{fig:2lenstri}). There are $\ell-2$ bits in the sequence, exactly half of which must be zeros and half of which must be ones, so the number of  these triangulations is
\[
\binom{\ell-2}{(\ell-2)/2}=\Omega\left(\frac{2^\ell}{\sqrt\ell}\right).
\]

\begin{figure}[t]
\centering\includegraphics[width=\textwidth]{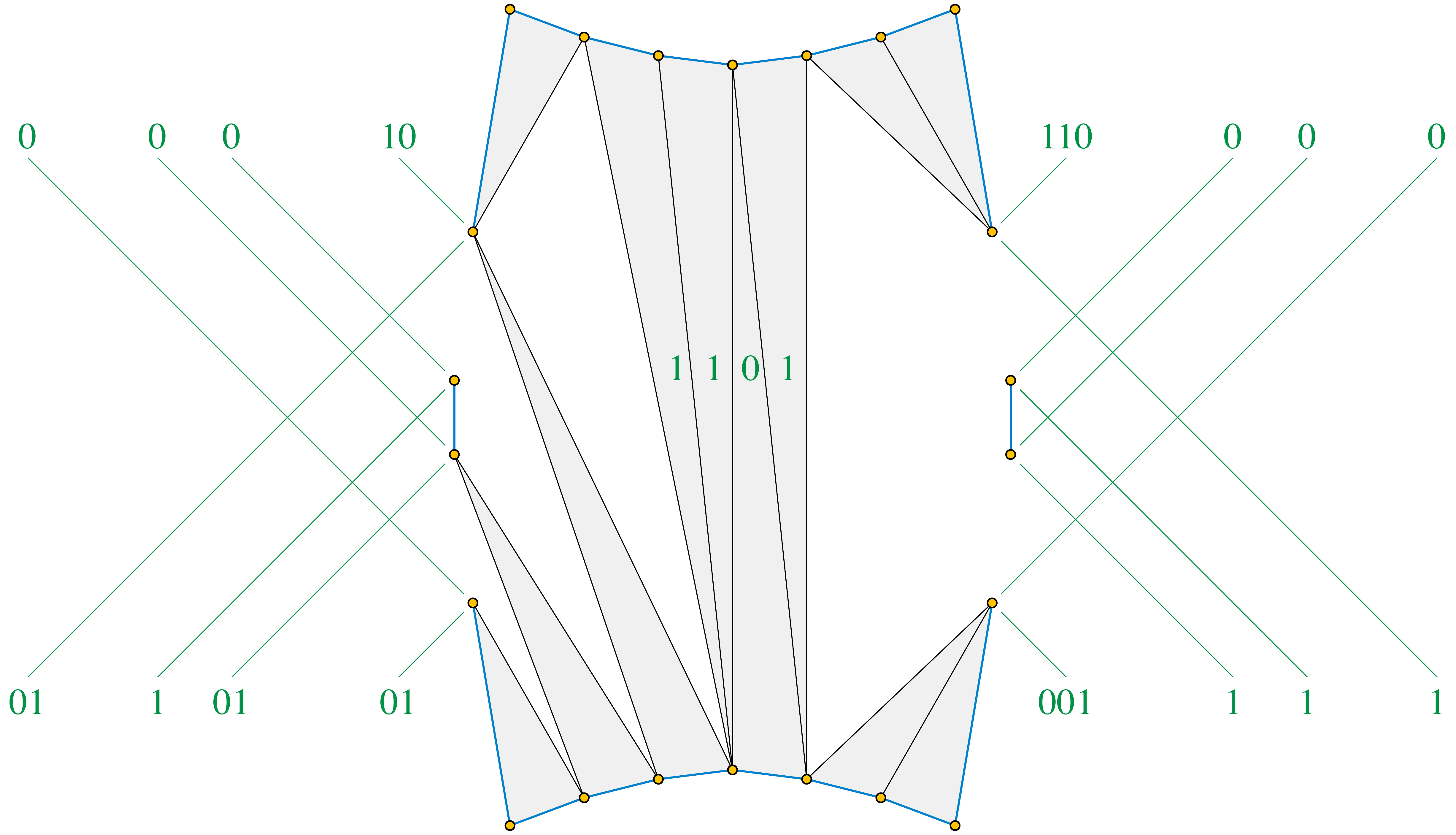}
\caption{Encoding the local part of a triangulation by a sequence of bits}
\label{fig:encode-local}
\end{figure}

In the other direction, we can produce an overestimate of the number of  distinct local parts of triangulations of $\Pi_S$ by a similar method of counting balanced binary strings of slightly greater length.
As in \autoref{fig:2lenstri}, consider one of the two lenses as being ``upper'' (labeled by 1's in the binary string) and the other as being ``lower'' (labeled by 0's). Similarly, we can describe the two sides of the tube of the gadget as being left or right. Suppose also that each side of a tube of the given gadget has $t$ internal vertices (vertices that do not belong to either lens; for our gadgets, $t\in\{4,6\}$). We describe the local part of a triangulation by a sequence of bits, as follows (\autoref{fig:encode-local}):
\begin{itemize}
\item $t$ blocks of 0-bits, each terminated by a single 1-bit, specifying how many edges of the lower lens form triangles whose apex is at each of the internal vertices of the left side of the tube.
\item $t$ blocks of 1-bits, each terminated by a single 0-bit, specifying how many edges of the upper lens form triangles whose apex is at each of the internal vertices of the left side of the tube.
\item A sequence of bits as before describing the left-to-right sequence of triangles that have all three vertices on the two lenses, with a 0-bit for a triangle with a side on the lower lens and a 1-bit for a triangle with a side on the upper lens.
\item $t$ blocks of 0-bits, each terminated by a single 1-bit, specifying how many edges of the lower lens form triangles whose apex is at each of the internal vertices of the right side of the tube.
\item $t$ blocks of 1-bits, each terminated by a single 0-bit, specifying how many edges of the upper lens form triangles whose apex is at each of the internal vertices of the right side of the tube.
\end{itemize}
The resulting sequence of bits uniquely describes each distinct local part, has $2\ell-2+4t$ bits, and has equal numbers of 0- and 1-bits. Not all such sequences of bits describe a valid local part (they may specify triangles connecting the lower and upper lenses to the tubes that cross each other) but this is non-problematic. There are $O(2^\ell/\sqrt\ell)$ possible sequences of bits of this type, so there are $O(2^\ell/\sqrt\ell)$ local parts of triangulations of $\Pi_S$. Each local part leaves remaining untriangulated regions on the left and right sides of $\Pi_S$ with $O(1)$ vertices (as in the proof of \autoref{lem:nonlocal}) so it can be extended to a complete triangulation in $O(1)$ distinct ways. Therefore, the total number of triangulations of $\Pi_S$ is $O(2^\ell/\sqrt\ell)$.
\end{proof}

\subsection{Global counting}

The key properties of the number of triangulations of $P_A$ (with an appropriate choice of $a$ and $b$) that allow us to prove our counting reduction are that 
\begin{itemize}
\item the number of triangulations coming from any particular maximum-size non-crossing subset of $A$ can be determined only from the size of the arrangement and the number of red segments in the subset, and
\item non-crossing subsets that use fewer segments than the maximum, or that use the maximum total number of segments but fewer than the maximum possible number of red segments, contribute a negligible fraction of the total number of triangulations.
\end{itemize}
In this section we make these notions precise.

\begin{figure}[t]
\centering\includegraphics[width=0.5\textwidth]{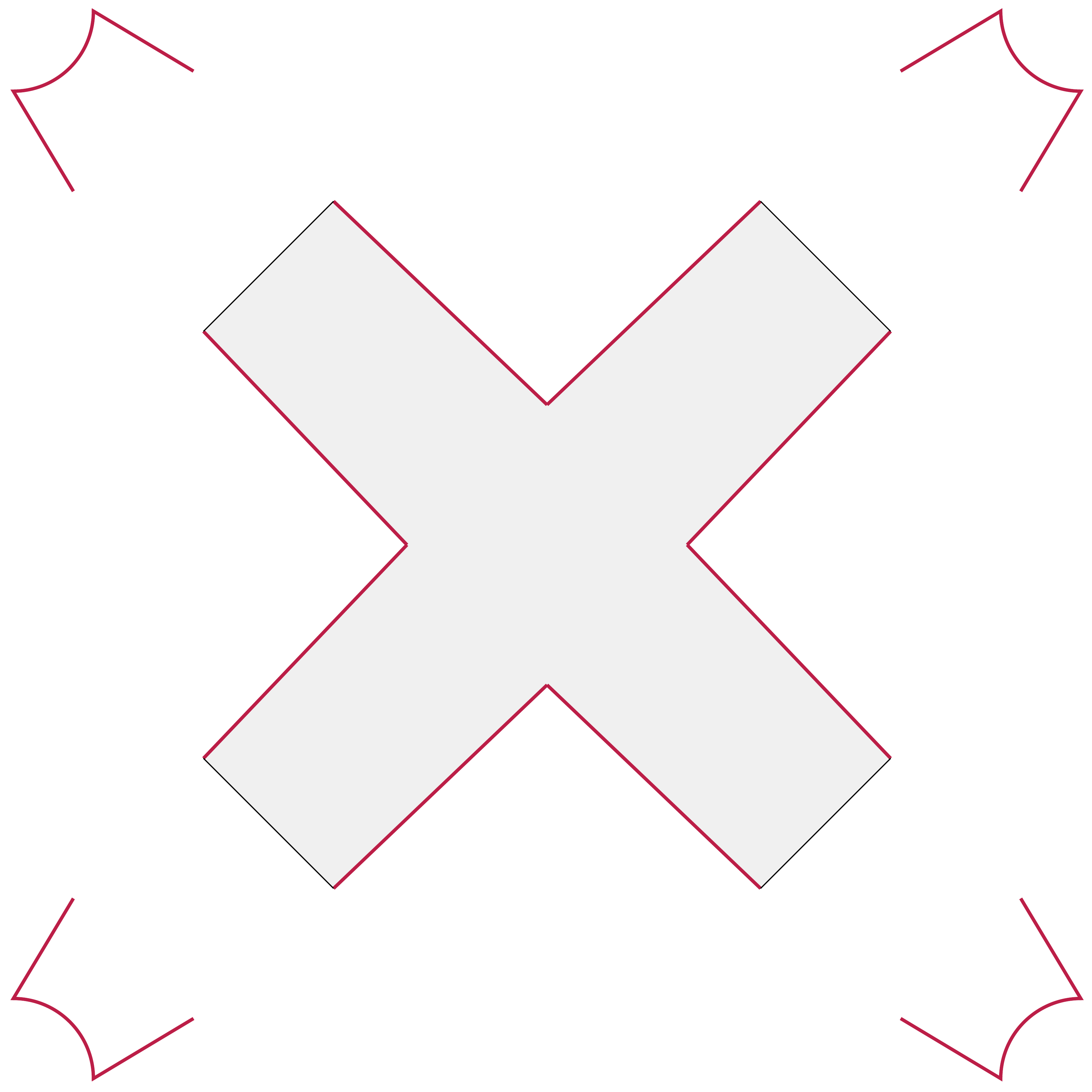}
\caption{The twelve-sided polygon in the center of two crossing red gadgets}
\label{fig:redcross}
\end{figure}

\begin{definition}
We define the following numerical parameters:
\begin{itemize}
\item $\alpha$ is the number of triangulations of polygon $\Pi_S$ (as defined in \autoref{def:polygons}) for  a blue gadget $S$, counting only the triangulations in which at least one triangle has three lens vertices. The definition of the gadgets causes gadgets of the same color to have equal numbers of triangulations.
\item $\beta$ is the number of triangulations of $\Pi_S$ for  a red gadget $S$, again counting only the triangulations in which at least one triangle has three lens vertices.
\item $\gamma$ is the number of triangulations of the twelve-sided polygon formed in the center of two crossing red gadgets by the four shared vertices of the two gadgets and their four neighbors in each gadget (\autoref{fig:redcross}). The constraints on how gadgets can cross imply that this number is independent of the choice of which two gadgets cross.
\item $q_{n,r}$, for any $n\ge 0$ and $0\le r\le n/2$, is $\alpha^{n-r} \beta^r \gamma^{(n-2r)/2} a^{2r} (a+b)^{2(n-r)} 2^{3r}$.
\end{itemize}
\end{definition}

The requirement in the definition of $\alpha$ and $\beta$ that the triangulations have a triangle with three lens vertices implies that these triangulations are exactly the ones that can appear in triangulations of $P_A$ for which the corresponding segment of $A$ is active. We do not derive an explicit formula for $\alpha$ and $\beta$, although they are bounded to within a constant factor by \autoref{lem:count-active}.  Therefore, it is important for us that they can be calculated in polynomial time by dynamic programming, as (for appropriate choices of $a$ and $b$) they count triangulations of simple polygons of polynomial size. (The constraint on the existence of a triangle with three lens vertices is not problematic for this dynamic programming algorithm.)

\begin{figure}[t]
\centering\includegraphics[width=0.65\textwidth]{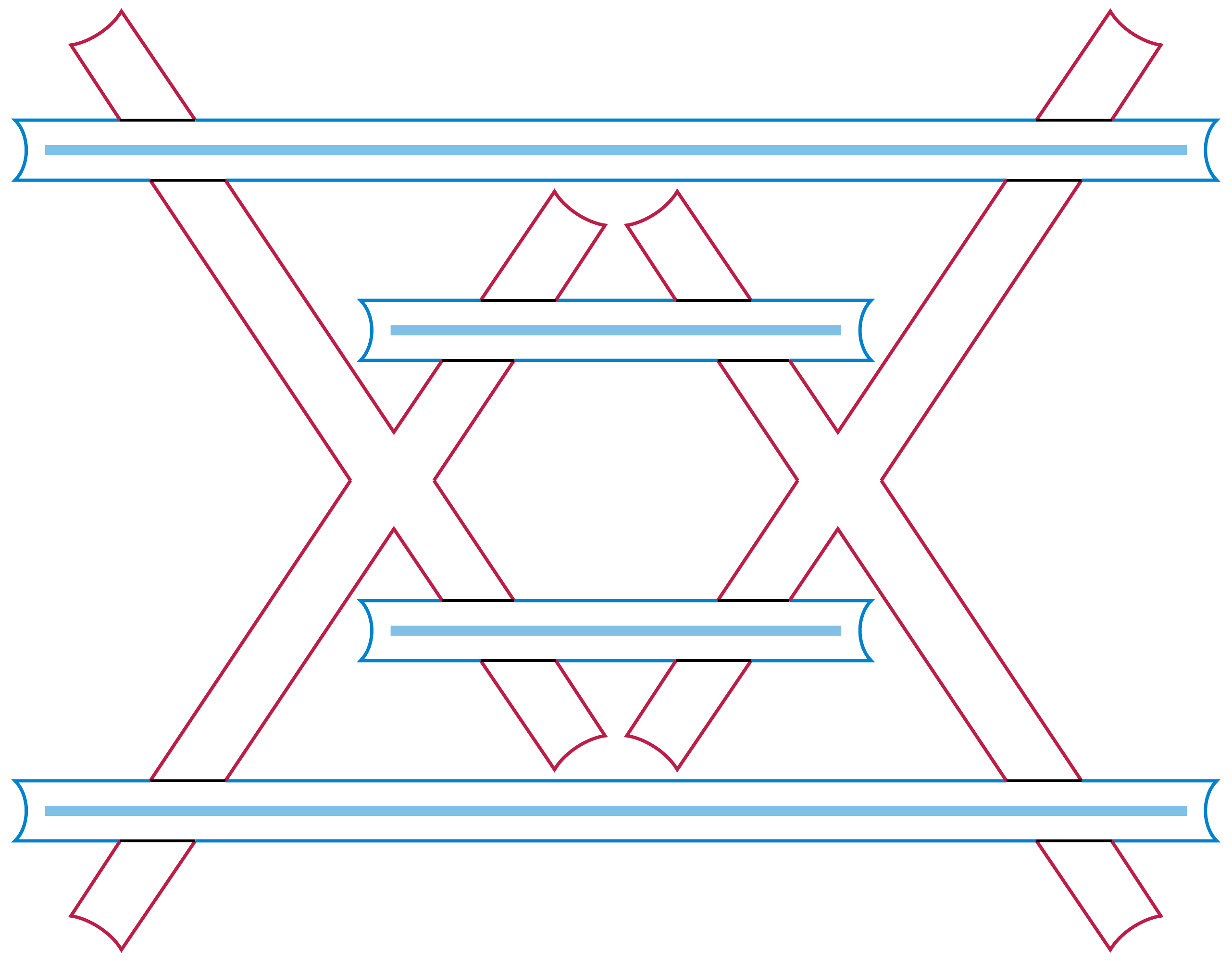}
\caption{A maximum-non-crossing subset of a red--blue arrangement (consisting of four blue segments and no red segments), and the associated partition of $P_A$ into smaller polygons (shown by the black segments). These smaller polygons include a polygon $\Pi_S$ for the gadget $S$ of each active segment $s$, two polygons $Q_L$ for the lenses of each gadget of a passive segment, and a twelve-sided polygon for each crossing of two passive red gadgets. Not shown: the central quadrilaterals of passive blue gadgets and of passive red gadgets crossed by an active red gadget.}
\label{fig:polypart}
\end{figure}

\begin{lemma}
\label{lem:exact-count}
Let $A$ be a red--blue arrangement with $n$ red and $n$ blue segments.
Let $\Xi$ be a maximum-size non-crossing subset of $A$, with $r=r(\Xi)$ red segments. Then the number of triangulations of $P_A$ in which $\Xi$ forms the active set is $q_{n,r}$.
\end{lemma}

\begin{proof}
By \autoref{lem:active-triangulation}, the triangulations with $\Xi$ active can be partitioned into the polygons $\Pi_S$ for the active segments of $\Xi$ and the remaining smaller polygons left by the removal of these polygons, each of which can be triangulated independently and each of which contributes a factor to the product in the definition of $q_{n,r}$. \autoref{fig:polypart} illustrates this partition into polygons. There are $r$ active red segments and $n-r$ active blue segments, each of which contributes a polygon $\Pi_S$ to the partition and a factor of $\beta$ or $\alpha$ (respectively) to the product by which $q_{n,r}$ was defined. There are $2r$ lenses of passive blue segments and $2(n-r)$ lenses of passive red segments, each of which contributes a polygon $Q_L$ containing the lens and a factor of $a$ or $(a+b)$ to the product, respectively, by \autoref{cor:passive-triangulation} and \autoref{lem:count-lenstri}.

The $r$ active red segments and $r$ passive red segments that they cross leave $n-2r$ red segments that are passive but not crossed by another active red segment. These $n-2r$ red segments form $(n-2r)/2$ twelve-sided polygons where their gadgets cross each other in pairs, with each pair contributing a factor of $\gamma$ to the product.

The remaining factor of $2^{3r}$ in the product comes from the $r$ passive blue gadgets, each of which has a central quadrilateral between the two red active segments that cross it, and from the $r$ passive red gadgets that are crossed by an active red gadget, each of which has two central quadrilaterals between the consecutive pairs of the three active segments that cross it. These $3r$ quadrilaterals each can be triangulated in two ways.
\end{proof}

\begin{corollary}
\label{cor:less-red}
Suppose that $a\ge b$. Then $q_{n,r}/q_{n,r-1}=\Theta(2^{2b})$.
\end{corollary}

\begin{proof}
Expanding the definition of $q_{n,r}$ and eliminating common factors with $q_{n,r-1}$ produces
\[
\frac{q_{n,r}}{q_{n,r-1}}
=
\frac{8\beta a^2}{\alpha\gamma(a+b)^2}
=
\Theta\Bigl(\frac{\beta}{\alpha}\Bigr).
\]
Here, $\gamma$ has been eliminated from the right hand side as a constant, and the assumption that $a\ge b$ implies that $1/4\le a^2/(a+b)^2\le 1$ allowing the elimination of those factors as well.
Applying \autoref{lem:count-active} to bound both $\alpha$ and $\beta$ to within constants
gives $\beta/\alpha=\Theta(2^{2b})$, the stated bound.
\end{proof}

\begin{lemma}
\label{lem:log-good}
Suppose the parameters $a$ and $b$ are both upper-bounded by polynomials of $n$. Then
\[
\log_2 q_{n,r} = 2an + 2rb \pm O(n\log n),
\]
where the constant in the $O$-notation depends on the bounds on $a$ and $b$.
\end{lemma}

\begin{proof}
In the formula for $q_{n,r}$, each factor of $\alpha$ contributes $2a-O(\log n)$ to the logarithm,
each factor of $\beta$ contributes $2a+2b-O(\log n)$, each factor of $a$ or $b$ contributes $O(\log n)$, and each factor of $\gamma$ or $2$ contributes $O(1)$. The result follows by adding these contributions according to their exponents.
\end{proof}

\begin{lemma}
\label{lem:log-bad}
Suppose the parameters $a$ and $b$ are both upper-bounded by polynomials of $n$.
Let $p$ denote the number of triangulations of $P_A$ that do not have a maximum-size non-crossing subset of $A$ as their active set. Then
\[
\log_2 p \le 2a(n-1)+2b(n-1) + O(n\log n).
\]
\end{lemma}

\begin{proof}
We bound $p$ by a product of the number of active sets, the numbers of local parts of triangulations for each gadget, and the numbers of ways of extending a choice of local parts to a full triangulation. This corresponds to bounding $\log_2 p$ by a sum of terms, one for each of these factors.

There are $2n$ segments of $A$, so at most $2^{2n}$ subsets of these segments that can be chosen as the active set, contributing a term of at most $2n$ to the logarithm.
By \autoref{lem:active-non-crossing}, the active set can be only a maximum-size non-crossing subset or a non-crossing subset of smaller than the maximum size $n$, and this lemma considers only the case where it is of smaller than maximum size, so each non-maximum active set can have at most $n-1$ active segments. By \autoref{lem:count-active}, the number of choices of local parts of gadgets for active segments adds at most $2a(n-1)+2b(n-1)-O(n\log n)$ to the logarithm. The local parts of the gadgets for passive segments can be chosen in a number of ways per gadget that is polynomial in $a$ and $b$ by \autoref{lem:passive-is-poly}, adding another $O(n\log n)$ term to the logarithm.
Once the active gadgets have been triangulated and the local parts of the passive gadgets have been chosen, the regions that remain to be triangulated have a total of $O(n)$ vertices by \autoref{lem:nonlocal}, so the number of ways to triangulate them contributes another $O(n)$ to the logarithm.
\end{proof}

\subsection{Completing the reduction}

We have already described how to transform an arrangement $A$ into a polygon $P_A$, modulo the choice of the parameters $a$ and $b$. To complete the description of our reduction, we need to set $a$ and $b$ and we need to describe how to recover the number of maximum-size non-crossing subsets of $A$ from the number $N$ of triangulations of $P_A$.

Given a red--blue arrangement $A$ with $n$ red and $n$ blue segments, set $b=2n$.
By \autoref{cor:less-red}, with this choice, $q_{n,r-1}$ and $q_{n,r}$ differ by a factor of $\Theta(2^{4n})$, much larger than the $2^{2n}$ bound on the total number of non-crossing sets of segments. Therefore, for all sufficiently large $n$ and all $r\le n$, the number of triangulations with a maximum-size non-crossing set of active segments that includes fewer than $r$ red segments is strictly less than $q_{n,r}$.
Next, set $a=3n^2$. This is large enough that, for sufficiently large $n$, the number $p$ of triangulations whose active set is non-maximum is strictly smaller than $q_{n,0}$. This follows because the bound of \autoref{lem:log-good} on $q_{n,0}$ has a larger multiple of $a$ than does the bound of \autoref{lem:log-bad} on $p$, and the $O(n\log n)$ term and larger multiple of $b$ in \autoref{lem:log-bad} are not large enough to make up for this difference.

With these choices of $a$ and $b$, we also have that $q_{n-1,n-1}$ is smaller than $q_{n,0}$ by a factor of $2^{\Omega(n^2)}$. Thus, for all sufficiently large $n$, if $A'$ is any red--blue arrangement with fewer than $n$ red and blue segments, the total number of triangulations of $P_{A'}$ is strictly smaller than $q_{n,0}$. This implies that, from the total number $N$ of triangulations of $P_A$ we can unambiguously determine the number of segments in~$A$.

With these considerations, we are ready to prove the correctness of our reduction:

\begin{proof}[Proof of \autoref{thm:main}]
Counting triangulations of a given polygon $P$ is in \sharpp{} by \autoref{lem:can-count},
so we can complete the proof by describing a polynomial-time counting reduction from the number of maximum-size non-crossing subsets of a red--blue arrangement $A$ (proved \sharpp-hard in \autoref{lem:red-blue-hard}).
To transform $A$ into a polygon, let $n$ be the number of red segments in $A$. Our reduction will work for all sufficiently large $n$, larger than some fixed threshold $n'$; the reduction algorithm compares $n$ to this threshold and, for $n<n'$, directly computes the number of maximum-size non-crossing subsets of $A$ (by a brute force search over all subsets, in time $O(2^{n'} (n')^{O(1)})=O(1)$) and constructs a polygon with the same number of triangulations using the formula for numbers of triangulations of lens polygons of \autoref{lem:count-lenstri}.
Otherwise, it chooses $a=3n^2$ and $b=2n$ as above, and constructs the polygon $P_A$.

To transform the number $N$ of triangulations of the given polygon into the number of maximum-size non-crossing subsets of $A$, first check whether this number is sufficiently small that it comes from our special-case construction for bounded values of $n$. If so, directly decode it to the number of maximum-size non-crossing subsets.

In the remaining case, recover $n$ as the unique value that could have produced a polygon $P_A$ with $N$ triangulations. From $n$, compute the associated parameters $a$ and $b$, use the dynamic programming algorithm for counting triangulations of simple polygons to compute the parameters $\alpha$ and $\beta$, and use these values to compute the parameters $q_{n,i}$. For each choice of $r$ from $n$ down to $0$ (in decreasing order)
compute the number of maximum-size non-crossing subsets with $r$ red segments as $\lfloor N'/q_{n,r}\rfloor$ (where $N'$ starts with the value $N$ and is reduced as the algorithm progresses) and then replace $N'$ by $N'$ mod $q_{n,r}$ before proceeding to the next value of $r$.
Sum the numbers of maximum-non-crossing subsets obtained for each value of $r$ to obtain the total number of maximum-size non-crossing subsets.

When computing the number of non-crossing subsets for each value of $r$, the contributions from triangulations whose active segments include more red segments than $r$ will already have been subtracted off, by induction.
The contributions from triangulations whose active segments include fewer red segments than $r$,
or from triangulations that do not have maximum-size non-crossing sets of active segments,
will sum to less than a single multiple of the number of triangulations for each non-crossing set with the given number of red segments, as discussed above. Therefore, each number of non-crossing subsets is computed correctly.
\end{proof}

\section{Conclusions and open problems}

We have shown that counting triangulations of polygons with holes is \sharpp-complete under Turing reductions. It would be of interest to tighten this result to show completeness under counting reductions, or even under parsimonious reductions. Can this be done, either by strengthening the type of reduction used for the underlying graph problem that we reduce from, independent sets in regular planar graphs, or by finding a different reduction for triangulations that bypasses the Turing reductions used for this graph problem?

In a triangulation of a polygon with holes, every hole has a diagonal connecting its leftmost vertex to a vertex to the left of it in another boundary component.
By testing all combinations of these left diagonals, and using dynamic programming to count triangulations of the simple polygon formed by cutting the input along one of these sets of diagonals (avoiding triangulations that use previously-tested diagonals) it is possible to count triangulations of an $n$-vertex polygon with $h$ holes in time $O(n^{h+3})$.  Is the dependence on $h$ in the exponent of $n$ necessary, or is there a fixed-parameter tractable algorithm for this problem?

More generally, there are many other counting problems in discrete geometry for which we neither know a polynomial time algorithm nor a hardness proof. For instance, we do not know the complexity of counting triangulations, planar graphs, non-crossing Hamiltonian cycles, non-crossing spanning trees, or non-crossing matchings of sets of $n$ points in the plane. Are these problems hard?

\section*{Acknowledgements}
A preliminary version of this paper appeared in the Proceedings of the 2019 International Symposium on Computational Geometry. This work was supported in part by the US National Science Foundation under grants  CCF-1618301 and CCF-1616248.

\bibliographystyle{amsplainurl}
\bibliography{counting}
\end{document}